\pdfoutput=1
\documentclass{llncs}
\usepackage{amsmath, amsfonts, amssymb, euscript}
\usepackage{stmaryrd}
\usepackage{textcomp}
\usepackage{geometry}
\usepackage{cite}
\usepackage{tikz}

\usepackage{times}

\usepackage{paralist}
\usepackage{cite}
\usepackage{color}
\usepackage{subfigure}
\usepackage{times}
\usepackage{epsfig}
\usepackage{graphics}
\usepackage{float}
\usepackage{color}
\usepackage{url}
\usepackage{paralist}
\usepackage{algorithm}
\usepackage{algpseudocode}
\usepackage{multirow}
\usepackage{tikz}
\usetikzlibrary{plotmarks,shapes,snakes}
\usepackage{pgfplots}
\usepackage{icomma}

\newcommand\p[1]{\fbox{#1}}

\usepackage{soul}

\newcommand{\set}[1]{\left\{#1\right\}}

\newcommand{\ignore}[1]{}

\DeclareMathOperator{\OPT}{OPT}
\DeclareMathOperator{\IB}{IB}
\DeclareMathOperator{\PO}{PO}
\DeclareMathOperator{\NPO}{NPO}
\DeclareMathOperator{\LPO}{LPO}

\definecolor{lightblue}{RGB}{200,200,255}
\definecolor{lightgreen}{RGB}{200,255,200}
\definecolor{darkgreen}{RGB}{100,155,100}

\tikzstyle{every node}=[font=\footnotesize]

\tikzstyle{gate}            = [circle,fill=white,draw=black,minimum size=8pt,minimum height=8pt,inner sep=2pt,font=\small]
\tikzstyle{block}           = [rectangle,fill=white,draw=black,minimum size=15pt,text width=280pt,inner sep=2pt,rounded corners,font=\small]
\tikzstyle{smallblock}      = [rectangle,fill=white,draw=black,minimum size=15pt,text width=80pt,inner sep=2pt,rounded corners,font=\small]
\tikzstyle{datablock}       = [rectangle,fill=lightblue,draw=black,minimum size=15pt,minimum width=200pt,inner sep=2pt,rounded corners,font=\small]
\tikzstyle{ifelse}          = [rectangle,fill=lightgreen,draw=black,minimum height=10pt,minimum width=200pt,inner sep=2pt,rounded corners,font=\small]
\tikzstyle{wire}            = [draw,thick,->]
\tikzstyle{ewire}           = [draw,thick]
\tikzstyle{edgelabel}       = [draw=white,fill=white,circle,inner sep=0pt]
\tikzstyle{enode}           = [rectangle,inner sep=2pt,rounded corners]

\newcommand{\tikzsimul}[2]{
\begin{tikzpicture}[scale=0.65]
\begin{axis}[xlabel={#2},ymin=0.15,ymax=1.1,height=5cm,width=8cm, scaled x ticks=true,
legend style={legend columns=4,at={(0.01,0.22)},anchor=north west,font=\footnotesize,draw=none},
ytick={0.2,0.4,0.6,0.8,1.0},
yticklabels={0.2,0.4,0.6,0.8,1.0},
x tick label style={font=\footnotesize}, y tick label style={font=\footnotesize} ]
#1
\legend{OPT*, PO, LPO, NPO}
\end{axis}
\end{tikzpicture}
}

\newcommand{\tikzsimulrightlegend}[2]{
\begin{tikzpicture}[scale=0.65]
\begin{axis}[xlabel={#2},ymin=0.15,ymax=1.1,height=5cm,width=8cm, scaled x ticks=true,
legend style={legend columns=2,at={(0.45,0.57)},anchor=north west,font=\footnotesize,draw=none},
ytick={0.2,0.4,0.6,0.8,1.0},
yticklabels={0.2,0.4,0.6,0.8,1.0},
x tick label style={font=\footnotesize}, y tick label style={font=\footnotesize} ]
#1
\legend{OPT*, PO, LPO, NPO}
\end{axis}
\end{tikzpicture}
}

\newcommand{\tikzsimulrightbottomlegend}[2]{
\begin{tikzpicture}[scale=0.65]
\begin{axis}[xlabel={#2},ymin=0.15,ymax=1.1,height=5cm,width=8cm, scaled x ticks=true,
legend style={legend columns=2,at={(0.45,0.37)},anchor=north west,font=\footnotesize,draw=none},
ytick={0.2,0.4,0.6,0.8,1.0},
yticklabels={0.2,0.4,0.6,0.8,1.0},
x tick label style={font=\footnotesize}, y tick label style={font=\footnotesize} ]
#1
\legend{OPT*, PO, LPO, NPO}
\end{axis}
\end{tikzpicture}
}

\title{FIFO Queueing Policies for Packets with Heterogeneous Processing}

\author{Kirill Kogan\inst{1}, Alejandro L\'opez-Ortiz\inst{1}, Sergey I. Nikolenko\inst{2,3}\thanks{Work of S.I.~Nikolenko, A.V.~Sirotkin, and D.~Tugaryov was supported by the Russian Fund for Basic Research grant 12-01-00450-a, the Russian Presidential Grant Programme for Young Ph.D.'s, grant no.~MK-6628.2012.1, for Leading Scientific Schools, grant no.~NSh-3229.2012.1, and RFBR grants~11-01-12135-ofi-m-2011 and 11-01-00760-a.}, Alexander V. Sirotkin\inst{4,3}${}^\star$ and Denis Tugaryov\inst{3}${}^\star$ }

\institute{
School of Computer Science, University of Waterloo
\and
Steklov Mathematical Institute, nab. r. Fontanka, 27, St. Petersburg, Russia
\and
St. Petersburg Academic University, ul. Khlopina, 8, korp. 3, St. Petersburg, Russia,
\and
St. Petersburg Institute for Informatics and Automation of the RAS, 14 Line VO, 39, St. Petersburg, Russia
}

\begin{document}

\maketitle

\begin{abstract}
We consider the problem of managing a bounded size First-In-First-Out (FIFO) queue
buffer, where each incoming unit-sized packet requires several rounds of processing before it can be transmitted out.
Our objective is to maximize the total number of successfully transmitted packets. We consider both push-out
(when the policy is permitted to drop already admitted packets) and non-push-out cases.
In particular, we provide analytical guarantees for the throughput performance of our algorithms.
We further conduct a comprehensive simulation study which
experimentally validates the predicted theoretical behaviour.

{{\bf Keywords}: Scheduling, Buffer Management, First-In-First-Out Queueing, Switches, Online Algorithms, Competitive Analysis.}
\end{abstract}

\section{Introduction}
\label{sec_introduction}

This work is mostly motivated by buffer management problems within Network Processors (NPs) in a packet-switched network.
Such NPs are responsible for complex packet processing tasks in modern high-speed routers,
including, to name just a few, forwarding, classification, protocol conversion, and intrusion detection.
Common NPs usually rely on multi-core architectures, where multiple cores perform various processing tasks required by the arriving traffic. Such architectures may be based on a pipeline of cores~\cite{xelerated}, a pool of identical cores~\cite{cavium,amcc,cisco}, or a hybrid pool pipeline~\cite{ezchip}.
In response to operator demands, packet processing needs are becoming more heterogeneous, as NPs need to cope with more complex tasks such as advanced VPN services and hierarchical classification for QoS, among others. Unlike general purpose processors, modern NPs employ run-to-completion processing. Recent results in data path
provisioning provide a possibility to have information about future required processing \emph{a priori} (for instance, this is possible in one of the modes of
the OpenFlow protocol \cite{OF}). In this work, we consider a model that captures the characteristics of this architecture. We evaluate the performance of such systems
for the case when information about required processing is available \emph{a priori}. The main concern in this setting is to maximize the throughput attainable by the NP, measured by the total number of packets successfully processed by the system.

In what follows, we adopt the terminology used to describe buffer management problems. We focus our attention on a general model where we are
required to manage admission control and scheduling modules of a single bounded size queue that process packets in First-In-First-Out order.
In this model, arriving traffic consists of unit-sized {\em packets}, and each packet has a {\em processing requirement} (in processor cycles).
A packet is successfully {\em transmitted} once the scheduling module has scheduled the packet for processing for at least its required
number of cycles.
If a packet is dropped upon arrival or pushed out from the queue after being admitted due to admission control policy considerations (if push-out is allowed), then the packet is lost without gain to the algorithm's throughput.

\subsection {Our Contributions}
\label{sec:our_contribution}
In this paper, we consider the problem of managing a FIFO queue buffer of size $B$, where each incoming unit-sized packet requires at most $k$ rounds of processing before it can be transmitted out. Our objective is to maximize the total number of successfully transmitted packets.
For online settings, we propose algorithms with provable performance guarantees.
We consider both push-out (when the algorithm can drop a packet from the queue) and non-push-out cases.
We show that the competitive ratio obtained by our algorithms depends on the maximum number of processing cycles required by a packet.
However, none of our algorithms needs to know the maximum number of processing cycles in advance.
We discuss the non-push-out case in Section~\ref{sec:simple} and show that the on-line greedy algorithm $\NPO$ is $k$-competitive, and
that this bound is tight.
For the push-out case, we consider two algorithms: a simple greedy algorithm $\PO$ that in the case of congestion pushes out the first packet with maximal required processing, and Lazy-Push-Out ($\LPO$) algorithm that mimics $\PO$ but does not transmit packets if there is still at least one admitted packet with more than one required processing cycle. Intuitively, it seems that $\PO$ should outperform $\LPO$ since $\PO$ tends to empty its buffer faster but we demonstrate that these algorithms are not comparable in the worst case. Although we provide a lower bound of $\PO$, the main result of this paper deals with the competitiveness of $\LPO$. In particular, we demonstrate that $\LPO$ is at most $\left(\ln k + 3 + \frac{o(B)}{B}\right)$-competitive. In addition, we demonstrate several lower bounds on the competitiveness of both $\PO$ and $\LPO$ for different values of $B$ and $k$. These results are presented in Section~\ref{sec:preemptive}. The competitiveness result of $\LPO$ is interesting in itself but since ``lazy'' algorithms provide a well-defined accounting infrastructure we hope that a similar approach can be applied to other systems in
 similar settings. From an implementation point of view we can define a new on-line algorithm that will emulate the behaviour of $\LPO$ and will not delay the transmission of processed packets. In Section~\ref{sec:simulation} we conduct a comprehensive simulation study to experimentally verify the performance of the proposed algorithms.
All proofs not appearing in the body of the paper can be found in the Appendix.

\subsection{Related Work}
\label{sec:related_work}

Keslassy et al.~\cite{KKSS+11} were the first to consider buffer
management and scheduling in the context of network processors with heterogeneous processing requirements for the arriving traffic.
They study both SRPT (shortest remaining processing time) and FIFO (first-in-first-out) schedulers with recycles,
in both push-out and non-push-out buffer management cases, where a packet is recycled after
processing according to the priority policy (FIFO or SRPT).
They showed competitive algorithms and worst-case lower bounds for such settings.
Although they considered a different architecture (FIFO with recycles) than the one we consider in this paper, they provided only a lower bound  for the push-out FIFO case, and it remains unknown if it can be attained.

Kogan et al.~\cite{KLSS+12} considered priority-based buffer management and scheduling in both push-out and non-push-out settings for heterogeneous packet sizes.
Specifically, they consider two priority queueing schemes: (i) Shortest Remaining Processing Time first (SRPT) and (ii) Longest Packet first (LP). They present
competitive buffer management algorithms for these schemes and provide lower bounds on the performance of algorithms for such priority queues.

The work of Keslassy et al.~\cite{KKSS+11} and Kogan et al.~\cite{KLSS+12}, as well as our current
work, can be viewed as part of a larger research effort concentrated on
studying competitive algorithms with buffer management for bounded buffers (see, e.g., a recent survey by
Goldwasser~\cite{G+10} which provides an excellent overview of this
field). This line of research, initiated in~\cite{MPL,KLMPSS-04},
has received tremendous attention in the past decade. 

Various models have been proposed and studied, including, among others, QoS-oriented models where packets have weights~\cite{MPL,KLMPSS-04,AielloMRR05,englert09lower} and
models where packets have dependencies~\cite{kesselman09competitive,mansour11overflow}.
A related field that has received much attention in recent years focuses on various
switch architectures and aims at designing competitive algorithms
for such multi-queue scenarios; see, e.g.,~\cite{AS+05,AR+06,AL+06,KKM+08,KKM+10}. Some other works also
provide experimental studies of these algorithms and further
validate their performance \cite{AJ+10}.

There is a long history of OS scheduling for multithreaded processors which is relevant to our research.
For instance, the SRPT algorithm has been studied extensively in such systems, and it is well known to be
optimal with respect to the mean response~\cite{scharge68proof}.
Additional objectives, models, and algorithms have been studied extensively in this context~\cite{LR+97,MRSG+05,MPT+94}. For a comprehensive overview of competitive online scheduling for server systems, see a survey by Pruhs \cite{K+07}.
When comparing this body of research with our proposed framework, one should note that OS scheduling is mostly concerned with average response time, but we focus on estimation of the throughput. Furthermore, OS scheduling does not allow jobs to be dropped, which is an inherent aspect of our proposed model since we have a limited-size buffer.

The model considered in our work is also closely related to job-shop scheduling problems~\cite{brucker06jobshop}, most notably
to hybrid flow-shop scheduling~\cite{ruiz10hybrid} in scenarios where machines have bounded buffers but are not allowed to drop and push out tasks.

\subsection{Model Description}

We consider a buffer with bounded capacity $B$ that handles
the arrival of a sequence of unit-sized packets. Each arriving packet $p$ is
branded with the number of required processing cycles $r(p) \in \set{1,\ldots,k}$.
This number is known for every arriving packet; for a motivation of why such information may be available see~\cite{WP+03}. Although the
value of $k$ will play a fundamental role in our analysis, we note that our algorithms need not know $k$ in
advance. In what follows, we adopt the terminology used in \cite{KLSS+12}. The queue performs two main tasks, namely {\em buffer
management}, which handles admission control of newly arrived
packets and push-out of currently stored packets, and {\em
scheduling}, which decides which of the currently stored packets
will be scheduled for processing. The scheduler will be determined
by the FIFO order employed by the queue. Our framework
assumes a multi-core environment, where we have $C$ processors, and
at most $C$ packets may be chosen for processing at any given
time. However, for simplicity, in the remainder of this paper we
assume the system selects a single packet for processing at any
given time (i.e., $C=1$).
This simple setting suffices to show both
the difficulties of the model and our algorithmic scheme.
We assume discrete slotted time, where each time slot $t$ consists of three phases:
\begin{enumerate}[(i)]
\item {\em arrival}: new packets arrive, and the buffer management unit performs admission control and, possibly, push-out;
\item {\em assignment and processing}: a single packet is selected for processing by the scheduling module;
\item {\em transmission}: packets with zero required processing left are transmitted and leave the queue.
\end{enumerate}

\begin{figure}
\begin{center}
\includegraphics[scale=0.76]{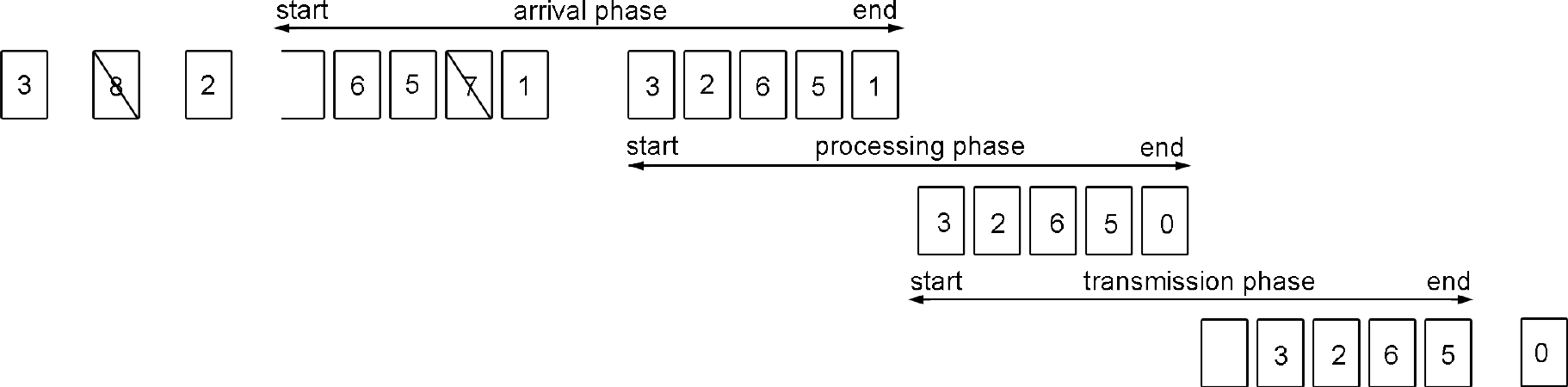}
\end{center}
\caption{Zoom in on a single time slot for greedy, push-out, and work-conserving algorithm.}
\label{fig:model}
\vspace{-0.5cm}
\end{figure}

If a packet is dropped prior to being {\em transmitted} (i.e., while it still has a positive number of
required processing cycles), it is lost. Note that a packet may be dropped either upon arrival
or due to a push-out decision while it is stored in the buffer. A packet contributes one unit to the
objective function only upon being successfully transmitted. The goal is to devise buffer management algorithms
that maximize the overall throughput, i.e., the overall number of packets transmitted from the queue.

We define a {\em greedy} buffer management policy as a
policy that accepts all arrivals if there is available
buffer space in the queue. A policy is {\em work-conserving} if it always processes whenever it has admitted packets
that require processing in the queue.

We say that an arriving packet $p$ {\em pushes out} a
packet $q$ that has already been accepted into the buffer
iff $q$ is dropped in order to free buffer space for $p$, and $p$ is admitted to the buffer instead in FIFO order.
A buffer management policy is called a {\em push-out} policy whenever it allows packets to push out currently stored packets. Figure~\ref{fig:model} shows a sample time slot in our model (for greedy and push-out case).

For an algorithm $ALG$ and a time slot $t$, we denote the set of packets stored in $ALG$'s buffer at time $t$ by $\IB^{ALG}_t$.

The number of {\em processing cycles} of a packet is key to
our algorithms. Formally, for every time slot $t$ and
every packet $p$ currently stored in the queue, its number of {\em residual
processing cycles}, denoted $r_t(p)$, is defined to be the number of
processing cycles it requires before it can be successfully
transmitted.

Our goal is to provide performance guarantees for various buffer management algorithms. We use competitive analysis~\cite{ST+85,Borodin-ElYaniv} when evaluating performance guarantees provided by our online algorithms. An algorithm $ALG$ is said to be {\em $\alpha$-competitive} (for some $\alpha \geq 1$) if for any arrival sequence $\sigma$ the number of packets successfully transmitted by $ALG$ is at least $1/\alpha$ times the number of packets successfully transmitted by an optimal
solution (denoted OPT) obtained by an offline clairvoyant algorithm.

\subsection{Proposed Algorithms}
\begin{algorithm}{}
\normalsize
\caption{{\sc $\NPO$}($p$): Buffer Management Policy}\label{alg:npo}
\begin{algorithmic}[1]
    \If {buffer occupancy is less than $B$}
        \State accept $p$
    \Else
        \State drop $p$
    \EndIf
\end{algorithmic}
\end{algorithm}

\vspace{-3mm}

\begin{algorithm}{}
\normalsize
\caption{{\sc $\PO$}($p$): Buffer Management Policy}
\label{alg:po}
\begin{algorithmic}[1]
   \If {buffer occupancy is less than $B$}
     \State accept $p$
   \Else
     \State let $q$ be the first (from HOL) packet with maximal number of residual processing
     \If {$r_t(p)<r_t(q)$}
       \State drop $q$ and accept $p$ according to FIFO order
     \EndIf
   \EndIf
\end{algorithmic}
\end{algorithm}
Next we define both push-out and non-push-out algorithms.
The Non-Push-Out Algorithm ($\NPO$) is a simple greedy work-conserving policy that accepts a packet if there is enough available buffer space.  Already admitted packets are processed in First-In-First-Out order. If during arrivals $\NPO$'s buffer is full then any arriving packet is dropped even if it has less processing required than a packet already admitted to $\NPO$'s buffer (see Algorithm~\ref{alg:npo}).

Next we introduce two push-out algorithms. The Push-Out Algorithm ($\PO$) is also greedy and work-conserving, but now, if an arriving packet $p$ requires less processing cycles than at least one packet in its buffer, then $\PO$ pushes out the first packet with the maximal number of processing cycles in its buffer and accepts $p$ according to FIFO order (see Algorithm~\ref{alg:po}).
The second algorithm is a new Lazy-Push-Out algorithm $\LPO$ that mimics the behaviour of $\PO$ with two important differences:
\begin{inparaenum}[(i)]
\item $\LPO$ does not transmit a Head-Of-Line packet with a single processing cycle if its buffer contains at least one packet
with more than one residual processing cycle, until the buffer contains only packets with a single residual processing cycle;
\item once all packets in $\LPO$'s buffer (say there are $m$ packets there) have a single processing cycle remaining,
$\LPO$ transmits them over the next $m/C$ processing cycles where $C$ is the number of processing cores;
observe that during this time, if an arriving packet $p$ requires less processing than the first packet $q$ with maximal number of processing cycles
in $\LPO$'s buffer, $p$ pushes out $q$ (similarly to $\PO$).
\end{inparaenum}

Intuitively, $\LPO$ is a weakened version of $\PO$ since $\PO$ tends to empty its buffer faster. Simulations also support this view (see Section~\ref{sec:simulation}).
However, Theorem~\ref{l:PO<LPO} shows that $\LPO$ and $\PO$ are incomparable in the worst case.

\begin{theorem}
\label{l:PO<LPO} 
\begin{inparaenum}[(1)]
\item There exists a sequence of inputs on which $\PO$ processes $\ge\frac{3}{2}$ times more packets than $\LPO$.
\item There exists a sequence of inputs on which $\LPO$ processes $\ge\frac{5}{4}$ times more packets than $\PO$.
\end{inparaenum}
\end{theorem}
\begin{proof}
To see (1), consider two bursts of $B$ packets:
\begin{itemize}
\item first burst of $B$ packets with required work $2$ arriving at time slot $t=1$;
\item second burst of $B$ packets with required work $1$ arriving at time slot $t=B$.
\end{itemize}
By the time the second burst arrives, $\LPO$ has processed no packets, while $\PO$ has processed $\frac B2$ packets.
Then both algorithms process $B$ packets of required work $1$ over the next $B$ time slots.
Since we have arrived at a state where both algorithms have empty buffers, we can repeat the procedure, getting an asymptotic bound.

To prove (2), suppose for simplicity that $k>\frac B2$.
The following table demonstrates the sequence of arrivals and the execution of both algorithms
($\# ALG$ denotes the number of packets processed by $ALG$ up to this time; no packets arrive during time slots not shown in the table).

\begin{tabular}{l|c|r|c|r|c}
$t$ & Arriving & $\IB^{\LPO}_t$ & \# $\LPO$ & $\IB^{\PO}_t$ & \# $\PO$ \\\hline
$1$ & $\p{2}\times B$ & $\p{2}$\ldots $\p{2}$ & $0$ & $\p{2}$\ldots $\p{2}$ & $0$ \\
$B$ & $\p{k}\times \frac B2$ & $\p{1}$\ldots $\p{1}$ & $0$ & $\p{k}$\ldots $\p{k}$ $\p{2}$\ldots $\p{2}$ & $B/2$ \\
$2B-1$ & none & $\p{1}$ & $B-1$ & $\p{k}$\ldots $\p{k}$ $\p{1}$ & $B-1$ \\
$2B$ & $\p{1}\times\frac B2$ & $\p{1}$\ldots $\p{1}$ & $B$ & $\p{1}$\ldots $\p{1}$ $\p{k}$\ldots $\p{k}$ & $B$ \\
$\frac{5B}2-1$ & none & $\p{1}$ & $3B/2-1$ & $\p{1}$\ldots $\p{1}$ $\p{k}$\ldots $\p{k-B/2+1}$ & $B$ \\
$\frac{5B}2$ & $\p{1}\times B$ & $\p{1}$ \ldots $\p{1}$ & $3B/2$ & $\p{1}$\ldots $\p{1}$ & $B$ \\
$\frac{7B}2$ & none & $\emptyset$ & $5B/2$ & $\emptyset$ & $2B$ \\
\end{tabular}

Similar to (1), we can repeat this sequence.
\qed
\end{proof}

$\LPO$ is an online push-out algorithm that obeys the FIFO ordering model, so its competitiveness is an interesting result by itself. But we believe this type of algorithms to be a rather promising direction for further study since they provide a well-defined accounting infrastructure that can be used for system analysis in different settings. From an implementation point of view we can define a new on-line algorithm that will emulate the behaviour of $\LPO$ but will not delay the transmission of processed packets. Observe that such an algorithm is not greedy. Although we will briefly discuss the competitiveness of an $\NPO$ policy and lower
bounds for $\PO$, in what follows $\NPO$ and $\PO$ will be mostly used as a reference for the simulation study.

\section{Competitiveness of the Non-Push-Out Policy}
\label{sec:simple}

The following theorem provides a tight bound on the worst-case performance of $\NPO$; its proof is given in the Appendix.

\begin{theorem}
\label{thm:NPO-tight-bound}
\begin{enumerate}[(1)]
\item For a sufficiently long arrival sequence, the competitiveness of $\NPO$ is at least $k$.
\item For a sufficiently long arrival sequence, the competitiveness of $\NPO$ is at most $k$.
\end{enumerate}
\end{theorem}

As demonstrated by the above results, the simplicity of
non-push-out greedy policies does have its price. In the
following sections we explore the benefits of introducing
push-out policies and provide an analysis of their performance.

\section{Competitiveness of Push-Out Policies}
\label{sec:preemptive}
In this section, we show lower bounds on the competitive ratio of $\PO$ and $\LPO$ algorithms and prove an upper bound for $\LPO$.

\subsection{Lower bounds}
In this part we consider lower bounds on the competitive ratio of $\PO$ and $\LPO$ for different values of $k$ and $B$.
Proofs of Theorems~\ref{thm:lb1} and~\ref{thm:po-lower-general} are given in the Appendix.
\begin{theorem}
\label{thm:lb1}
The competitive ratio of both $\LPO$ and $\PO$ is at least $2\left(1-\frac1B\right)$ for $k\ge B$.
The competitive ratio for $k < B$ is at least $\frac{2k}{k+1}$ for $\PO$ and at least $\frac{2k-1}{k}$ for $\LPO$.
\end{theorem}

For large $k$ (of the order $k\approx B^n$, $n>1$), logarithmic lower bounds follow.

\begin{theorem}
\label{thm:po-lower-general}
The competitive ratio of $\PO$ ($\LPO$) is at least $\lfloor\log_Bk\rfloor + 1 - O(\frac1B)$.
\end{theorem}

\subsection{Upper Bound on the Competitive Ratio of $\LPO$}
We already know that the performance of $\LPO$ and $\PO$ is incomparable in the worst case (see Theorem~\ref{l:PO<LPO}), and
it remains an interesting open problem to show an upper bound on the competitive ratio of $\PO$.
In this section we provide the first known upper bound of $\LPO$. Specifically, we prove the following theorem.

\begin{theorem}
\label{thm:lpo_upper_bound}
$\LPO$ is at most $\left(\ln k + 3 + \frac{o(B)}{B}\right)$-competitive.
\end{theorem}

We remind that $\LPO$ does not transmit any packet until all packets in the buffer have exactly one processing cycle left.
The definition of $\LPO$ allows for a well-defined accounting infrastructure. In particular, $\LPO$'s definition helps us to define an \emph{iteration}
during which we will count the number of packets transmitted by the optimal algorithm and compare it to the
contents of $\LPO$'s buffer. The first iteration begins with the first arrival. An iteration ends when all packets in the $\LPO$ buffer have a single processing pass left. Each subsequent iteration starts after the transmission of all $\LPO$ packets from the previous iteration.

We assume that $\OPT$ never pushes out packets and it is work-conserving; without loss of generality, every optimal algorithm can be
assumed to have these properties since the input sequence is available for it a priori. Further, we enhance $\OPT$ with two additional properties: (1) at the start of each iteration, $\OPT$ flushes out all packets remaining in its buffer from the previous iteration (for free, with extra gain to its throughput);
(2) let $t$ be the first time when $\LPO$'s buffer is congested during an iteration; $OPT$ flushes out all packets that currently reside in its buffer at time $t-1$ (again,
for free, with extra gain to its throughput).
Clearly, the enhanced version of $\OPT$ is no worse than the optimal algorithm since both properties provide additional advantages to $\OPT$ versus the original optimal algorithm. In what follows, we will compare $\LPO$ with this enhanced version of $\OPT$ for the purposes of an upper bound.

To avoid ambiguity for the reference time, $t$ should be interpreted as the arrival time of a single packet. If more than one packet arrive at the same time slot,
this notation is considered for every packet independently, in the sequence in which they arrive (although they might share the same actual time slot).

\begin{claim}
\label{l:end}
Consider an iteration $I$ that begins at time $t'$ and ends at time $t$. The following statements hold:
\begin{inparaenum}[(1)]
\item \label{l:end:1} during $I$, the buffer occupancy of $\LPO$ is at least the buffer occupancy of $\OPT$;
\item \label{l:end:2} between two subsequent iterations $I$ and $I'$, $\OPT$ transmits at most $|\IB^{\LPO}_t|$ packets;
\item \label{l:end:3} if during a time interval $[t',t"]$, $t'\leq t"\leq t$, there is no congestion then during $[t',t"]$ $\OPT$ transmits at most $|\IB^{\LPO}_{t"}|$ packets.
\end{inparaenum}
\end{claim}
\begin{proof}
\begin{inparaenum}[(1)]
\item $\LPO$ takes as many packets as it can until its buffer is full and once full it remains so
      for the rest of the iteration hence its buffer is at least as full as $\OPT$'s during an
      iteration.
\item By (\ref{l:end:1}), at the end of an iteration the buffer occupancy of $\LPO$ is at least the buffer occupancy of $\OPT$; moreover, all packets in $\LPO$ buffer at the end of an iteration have a single processing cycle.
\item Since during $[t',t"]$ there is no congestion and since $\LPO$ is greedy, $\LPO$ buffer contains all packets that have arrived after $t$, and thus, $\OPT$ cannot transmit more packets than have arrived.
\end{inparaenum}
\qed
\end{proof}

We denote by $M_t$ the maximal number of residual processing cycles among all packets in $\LPO$'s buffer at time $t$;
by $W_t$, the total residual work for all packets in $\LPO$'s buffer at time $t$.

\begin{lemma}
\label{l:max}
For every packet accepted by $\OPT$ at time $t$ and processed by $\OPT$ during the time interval
$[t_s,t_e]$, $t\leq t_s\leq t_e$, if $|\IB^{\LPO}_{t-1}|=B$
then $W_{t_e}\leq W_{t-1}-M_{t}$.
\end{lemma}
\begin{proof}
If $\LPO$'s buffer is full then a packet $p$ accepted by $\OPT$ either pushes out a packet in $\LPO$'s buffer or is rejected by $\LPO$.
If $p$ pushes a packet out, then the total work
$W_{t-1}$ is reduced by $M_{t}-r_{t}(p)$. Moreover, after processing $p$, $W_{t_e}\leq W_{t-1}-(M_{t}-r_t(p))-r_t(p)= W_{t-1}-M_{t}$.
Otherwise, if $p$ is rejected by $\LPO$ then $r_{t}(p)\geq M_{t}$, and thus $W_{t_e}\leq W_{t-1}-r_t(p)\leq W_{t-1}-M_{t}$.\qed
\end{proof}

Let $t$ be the time of the first congestion during an iteration $I$ that has ended at time $t'$.
Observe that by definition, at time $t$, $OPT$ flushes out all packets that were still in its buffer at time $t-1$.
We denote by $f(B,W)$ the maximal number of packets that $\OPT$ can process during $[t,t']$, where $W=W_{t-1}$.

\begin{lemma}
\label{l:crux}
For every $\epsilon>0$, $f(B,W) \le \frac{B-1}{1-\epsilon}\ln\frac{W}{B} + o(B\ln\frac WB)$.
\end{lemma}
\begin{proof}
By definition, $\LPO$ does not transmit packets during an iteration. Hence, if the buffer of $\LPO$ is full, it will remain full until the end of iteration.
At any time $t$, $M_{t}\ge\frac{W_{t}}{B}$: the maximal required processing is no less than the average.
By Lemma~\ref{l:max}, for every packet $p$ accepted by $\OPT$ at time $t$, the total work $W=W_{t-1}$ is reduced by $M_t$ after $\OPT$ has processed $p$.
Therefore, after $\OPT$ processes a packet at time $t'$, $W_{t'}$ is at most $W\left(1-\frac{1}{B}\right)$.

We now prove the statement by induction on $W$. The base is trivial for $W=B$ since all packets are already $1$'s.

The induction hypothesis is now that after one packet is processed by $\OPT$, there cannot be more than
$f(B,\frac{W}{B}\left(1-\frac1B\right))\le \frac{B-1}{1-\epsilon}\ln\left[\frac{W}{B}\left(1-\frac1B\right)\right]$
packets left, and for the induction step we have to prove that
$$\frac{B-1}{1-\epsilon}\ln\left[\frac{W}{B}\left(1-\frac1B\right)\right]+ 1 \le \frac{B-1}{1-\epsilon}\ln\frac{W}{B} .$$
This is equivalent to
$$\ln\frac{W}{B} \ge \ln\left[\frac{W}{B}\frac{B-1}Be^{\frac{1-\epsilon}{B-1}}\right],$$
and this holds asymptotically because for every $\epsilon>0$, we have $e^{\frac{1-\epsilon}{B-1}}\le \frac{B}{B-1}$ for $B$ sufficiently large.\qed
\end{proof}

Now we are ready to prove Theorem~\ref{thm:lpo_upper_bound}.
\begin{proof}[of Theorem~\ref{thm:lpo_upper_bound}]
Consider an iteration $I$ that begins at time $t'$ and ends at time $t$.
\begin{enumerate}
\item \emph{$\LPO$'s buffer is not congested during $I$.}
In this case, by Claim~\ref{l:end}(\ref{l:end:3}) $\OPT$ cannot transmit more than $|\IB^{\LPO}_t|$ packets during $I$.
\item \emph{During $I$, $\LPO$'s buffer is first congested at time $t''$, $t'\leq t''\leq t$.}
If during $I$ $\OPT$ transmits less than $B$ packets then we are done. By Claim~\ref{l:end}(\ref{l:end:3}), during $[t',t'']$ $\OPT$ can transmit at most $B$ packets. Moreover, at most $B$ packets are left in $\OPT$ buffer at time $t''-1$. By Lemma~\ref{l:crux}, during $[t'',t]$ $\LPO$ transmits at most
$(\ln k + \frac{o(B)}{B})B$ packets (because $W\le kB$), so the total amount over a congested iteration is at most $(\ln k +2+\frac{o(B)}{B})B$ packets.
\end{enumerate}
Therefore, during an iteration $\OPT$ transmits at most $(\ln k +2+o(1))|\IB^{\LPO}_t|$ packets. Moreover, by Claim~\ref{l:end}(\ref{l:end:2}), between two subsequent iterations $\OPT$ can transmit at most $|\IB^{\LPO}_t|$ additional packets. Thus, $\LPO$ is at most $\ln k + 3 + \frac{o(B)}{B}$-competitive.\qed
\end{proof}

The bound shown in Theorem~\ref{thm:lpo_upper_bound} is asymptotic. To cover small values of $B$, we show a weaker bound ($\log_2{k}$ instead of $\ln k$) on inputs where $\LPO$ never pushes out packets that are currently being processed. 

The following theorem shows an upper bound for this family of inputs; it also provides motivation for a new algorithm that does not push out packets that are currently being processed. This restriction is practical (if a packet is being processed, perhaps this means that it has left the queue and gone on, e.g., to CPU cache), and the analysis of such an algorithm is an interesting problem that we leave open.

\begin{theorem}
\label{thm:upper_bound_for_any_B}
For every $B>0$ and $k>0$, if $\LPO$ never pushes out packets that are currently being processed then $\LPO$ is at most $\left(\log_2{k}+3+\frac{B-1}{B}\right)$-competitive.
\end{theorem}
\begin{proof}
The case when there is no congestion during iteration is identical to the same case of Theorem~\ref{thm:lpo_upper_bound}.

If, during an iteration $I$, $\LPO_{\mathrm{p}}$'s buffer is congested, it is full and it will remain full till the end of iteration.
If during $I$ $OPT$ transmits less than $B$ packets then we are done. Otherwise, consider sub-intervals of time during $I$ when $OPT$ transmits exactly $B$ packets. 

We denote by $A^{O}_i$ the average number of processing passes between all packets transmitted by $\OPT$ during the $i^{\text{th}}$ subinterval.
We also denote by $A^s_i$ and by $A^e_i$ the average number of residual processing passes among all packets in $\LPO_{\mathrm{p}}$'s buffer at the start and at the end
of the $i^{\text{th}}$ subinterval, respectively. Since any packet processed during the subinterval is not pushed out, $A^e_i=\min(A_i,A^s_i-A_i)$. Clearly, as a
result the maximal number of subintervals during an iteration is achieved when $A_i=A^s_i-A_i$. Therefore, the maximal number of subintervals during an iteration is
bounded by $\log_2{k}$ (recall that $A^s_1\leq k$). By definition, $\OPT$ can gain at most $2B$ packets at the time of the first congestion during iteration.
In the worst case, from the end of the last subinterval till the end of iteration $\OPT$ can transmit at most $B-1$ additional packets.  
Thus, during a congested iteration $\OPT$ transmits at most $(\log_2{k}+2)B+B-1$ packets. Moreover, by Claim~\ref{l:end}(\ref{l:end:2}), between two subsequent iterations $\OPT$ can transmit at most $B$ additional packets. Thus, $\LPO_{\mathrm{p}}$ is at most $\left(\log_2{k} + 3 + \frac{B-1}{B}\right)$-competitive.\qed
\end{proof}

\section{Simulation Study}
\label{sec:simulation}

In this section, we consider the proposed policies (both push-out and non-push-out) for FIFO buffers
and conduct a simulation study in order to further explore and validate their performance.
Namely, we compare the performance of $\NPO$, $\PO$, and $\LPO$ in different settings. It was shown in \cite{KKSS+11} that a push-out algorithm that processes packets with less required processing first is optimal. In what follows we denote it by $\OPT^*$. Clearly, $\OPT$ in the FIFO queueing model does not outperform $\OPT^*$.

Our traffic is generated using an ON-OFF Markov modulated
Poisson process (MMPP), which we use to simulate bursty traffic. The choice of parameters is governed by the
average arrival load, which is determined by the product of
the average packet arrival rate and the average number of
processing cycles required by packets. For a choice of parameters
yielding an average packet arrival rate of $\lambda$, where
every packet has its required number of passes chosen
uniformly at random within the range $[1,k]$, we obtain an
average arrival load (in terms of required passes) of $\lambda
\cdot \frac{k+1}{2}$.

In our experiments, the ``OFF'' state has average arrival rate $\lambda=0{.}3$, and the ``ON'' state has average arrival rate $\lambda=4{.}5$
(the number of packets is uniformly distributed between $3$ and $6$). By performing simulations for variable
values of the maximal number of required passes $k$ in the range $[1,40]$, we essentially evaluate the performance of
our algorithms in settings ranging from underload (average arrival load of $0{.}3$ for $k=1$ and $0{.6}$ for $k=2$)
to extreme overload (average arrival load of $180$ in the ``ON'' state for $k=40$), which enables us to validate the
performance of our algorithms in various traffic scenarios.

\begin{figure}
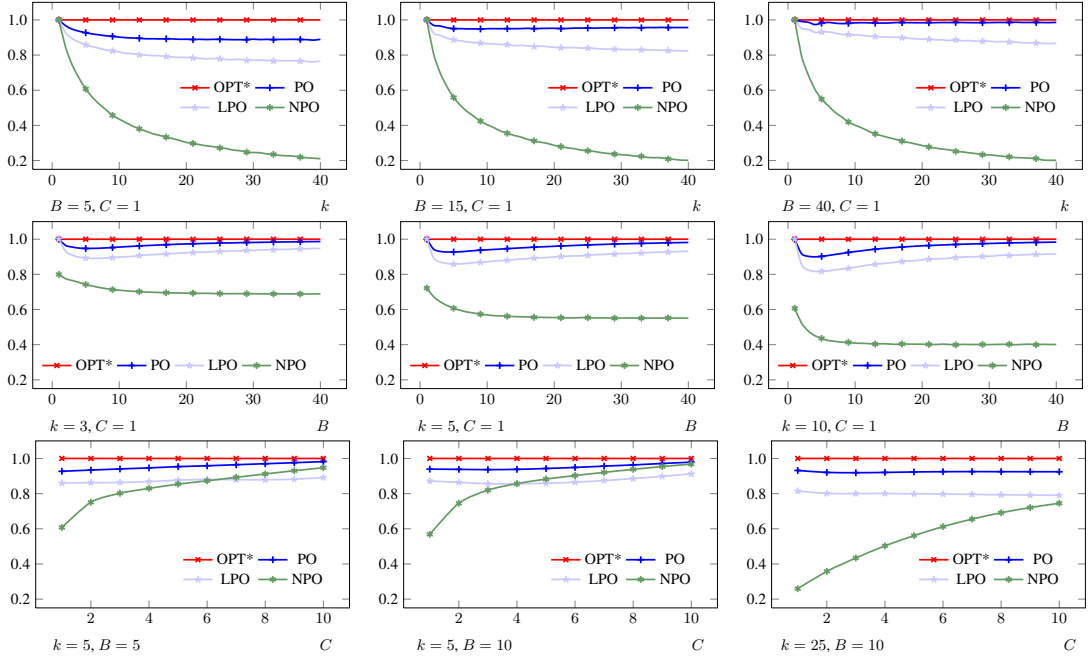

\begin{center}

\tikzsimulrightlegend{ \input{graphs/OPTSort.B5C1.mmpp.tex}
\input{graphs/FIFO.B5C1.mmpp.tex}
\input{graphs/LPQ.B5C1.mmpp.tex}
\input{graphs/GreedyNP.B5C1.mmpp.tex} }{$B=5$, $C=1$ \hspace{100pt} $k$}

\tikzsimulrightlegend{ \input{graphs/OPTSort.B15C1.mmpp.tex}
\input{graphs/FIFO.B15C1.mmpp.tex}
\input{graphs/LPQ.B15C1.mmpp.tex}
\input{graphs/GreedyNP.B15C1.mmpp.tex} }{$B=15$, $C=1$ \hspace{100pt} $k$}

\tikzsimulrightlegend{ \input{graphs/OPTSort.B40C1.mmpp.tex}
\input{graphs/FIFO.B40C1.mmpp.tex}
\input{graphs/LPQ.B40C1.mmpp.tex}
\input{graphs/GreedyNP.B40C1.mmpp.tex} }{$B=40$, $C=1$ \hspace{100pt} $k$}

\tikzsimul{ \input{graphs/OPTSort.k3C1.mmpp.tex}
\input{graphs/FIFO.k3C1.mmpp.tex}
\input{graphs/LPQ.k3C1.mmpp.tex}
\input{graphs/GreedyNP.k3C1.mmpp.tex} }{$k=3$, $C=1$ \hspace{100pt} $B$}

\tikzsimul{ \input{graphs/OPTSort.k5C1.mmpp.tex}
\input{graphs/FIFO.k5C1.mmpp.tex}
\input{graphs/LPQ.k5C1.mmpp.tex}
\input{graphs/GreedyNP.k5C1.mmpp.tex} }{$k=5$, $C=1$ \hspace{100pt} $B$}

\tikzsimul{ \input{graphs/OPTSort.k10C1.mmpp.tex}
\input{graphs/FIFO.k10C1.mmpp.tex}
\input{graphs/LPQ.k10C1.mmpp.tex}
\input{graphs/GreedyNP.k10C1.mmpp.tex} }{$k=10$, $C=1$ \hspace{100pt} $B$}

\tikzsimulrightbottomlegend{ \input{graphs/OPTSort.k5B5.mmpp.tex}
\input{graphs/FIFO.k5B5.mmpp.tex}
\input{graphs/LPQ.k5B5.mmpp.tex}
\input{graphs/GreedyNP.k5B5.mmpp.tex} }{$k=5$, $B=5$ \hspace{100pt} $C$}

\tikzsimulrightbottomlegend{ \input{graphs/OPTSort.k5B10.mmpp.tex}
\input{graphs/FIFO.k5B10.mmpp.tex}
\input{graphs/LPQ.k5B10.mmpp.tex}
\input{graphs/GreedyNP.k5B10.mmpp.tex} }{$k=5$, $B=10$ \hspace{100pt} $C$}

\tikzsimulrightbottomlegend{ \input{graphs/OPTSort.k25B10.mmpp.tex}
\input{graphs/FIFO.k25B10.mmpp.tex}
\input{graphs/LPQ.k25B10.mmpp.tex}
\input{graphs/GreedyNP.k25B10.mmpp.tex} }{$k=25$, $B=10$ \hspace{100pt} $C$}
\end{center}
\caption{Competitive ratio as a function of parameters: top row, of $k$; middle row, of $B$; bottom row, of $C$.}
\label{fig:simul}
\vspace{-0.5cm}
\end{figure}

Figure~\ref{fig:simul} shows the results of our simulations. The vertical axis in all
figures represents the ratio between the algorithm's
performance and $OPT^*$ performance given the
arrival sequence (so the red line corresponding to $\OPT^*$ is always horizontal at $1$).

We conduct three sets of simulations: the first one is targeted at a better
understanding of the dependence on the number of processing cycles, the second
evaluates dependency of performance from buffer size, and the third
aims to evaluate the power of having multiple cores.

We note that the standard deviation
throughout our simulation study never exceeds $0.05$
(deviation bars are omitted from the figures for readability).
For every choice of parameters, we conducted $200{,}000$ rounds (time slots) of simulation.

\subsection{Variable Maximum Number of Required Processing Cycles}
In these simulations, we restricted our attention to the single core case ($C=1$).
The top row of graphs on Fig.~\ref{fig:simul} shows that $\OPT^*$ keeps outperforming $\LPO$ and $\NPO$ more and more as $k$ grows.
In these settings, the difference in the order of processing between $OPT^*$ and $PO$ is small. The performance of $\LPO$ versus $\NPO$
degrades moderately since $\LPO$ is a push-out algorithm. This behaviour is of course as expected.

\subsection{Variable Buffer Size}
In this set of simulations we evaluated the performance of our algorithms for variable values of $B$ in the range $[1,40]$.
Throughout our simulations we again assumed a single core ($C=1$) and evaluated different values of $k$.
The middle row on Fig.~\ref{fig:simul} presents our results. Unsurprisingly, the performance of all algorithms
significantly improves as the buffer size increases; the difference between $\OPT^*$ and two other push-out algorithms visibly reduces,
but, of course, it would take a huge buffer for $\NPO$ to catch up (one would need to virtually remove the possibility of congestion).

\subsection{Variable Number of Cores}
In this set of simulations we evaluated the performance of our algorithms for variable values of $C$ in the range $[1,10]$.
The bottom row of Fig.~\ref{fig:simul} presents our results; the performance of all algorithms, naturally, improves drastically as the number of cores increases.
There is an interesting phenomenon here: push-out capability becomes less important since buffers are congested less often, but
$\LPO$ keeps paying for its ``laziness''; so as $C$ grows, eventually $\NPO$ outperforms $\LPO$.
The increase in the number of cores essentially provides the network processor (NP) with a speedup proportional to the number of cores (assuming the average arrival rate remains constant).

\section{Conclusion}
\label{sec:discussion}
The increasingly heterogeneous needs of NP traffic processing
pose novel design challenges for NP architects. In this paper, we provide
performance guarantees for NP buffer scheduling algorithms with FIFO queueing for
packets with heterogeneous required processing. The objective is to maximize the
number of transmitted packets under various settings such as push-out and
non-push-out buffers. We validate our results by simulations. As future
work, it will be interesting to show an upper bound for the $\PO$ algorithm and try
to close the gaps between lower and upper bounds of the proposed on-line algorithms.

\bibliographystyle{splncs}
\bibliography{np}

\newpage

\begin{appendix}
\section*{Appendix}
\label{sec:appendix}

\begin{proof}[of Theorem~\ref{thm:NPO-tight-bound}]
(\emph{1}) \emph{Lower bound}. Assume for simplicity that $B/C$ is an integer.
Consider the following set of arrivals. During the first time slot
arrive $B-C$ packets with maximal number of passes $k$. From the same time slot an iteration is started. During each iteration $C$ packets with $k$ processing cycles arrive. Since $\NPO$ is greedy it accepts all of them and its buffer is full. $OPT$
does not accept these packets. During the arrival phase of the next $kC$ time slots the buffer of $\NPO$ is full since it is non-push-out and implements FIFO order. During this time interval arrive $C$ packets with a single processing cycle each time slot. During each iteration $\OPT$ transmits $kC$ packets but $\NPO$ transmits only $C$ packets. In contrast to the other iterations during the last time slot of the last iteration a burst of $B$ packets arrives. So at the end of the arrival sequence both buffers are full. Thus, the competitiveness of $\NPO$ is at least $\frac{ikC+B}{iC+B}$, $i\geq 1$. Therefore, for sufficiently big value of $k$, $\NPO$ is at least $k$-competitive.

(\emph{2}) \emph{Upper bound}. Observe that $\NPO$ must fill up its buffer before it drops any packets.
Moreover, so long as the $\NPO$ buffer is not empty then after at most $k$ time steps $\NPO$
must transmit its HOL packet. This means that $\NPO$ is transmitting at a
rate of at least one packet every $k$ time steps, while $\OPT$ in the same
time interval transmitted at most $k$ packets.
Hence, the number of transmitted packets at time $t$ for $\NPO$ is at least
$\lfloor t/k\rfloor$ while $\OPT$ transmitted at most $t$ packets for a competitive ratio of $k$ so long as the $\NPO$ buffer did not become empty before $\OPT$'s did.

If, on the other hand, $\NPO$ empties its buffer first, this means there
were no packet arrivals since the $\NPO$ buffer went below the $B-1$
threshold at a time $t$. From that moment on $\NPO$ empties its buffer
transmitting thus at least $B-1$ packets, while $\OPT$ transmitted at most
$B$ packets.

So in total the number of packets transmitted by $\NPO$ is at least $ \left\lfloor\frac{t}{k}\right\rfloor+B-1$
while the total number of packets transmitted by $\OPT$ is $t + B$. Thus, for sufficiently long input sequences $\NPO$ is $k$-competitive.
\qed
\end{proof}

\begin{proof}[of Theorem~\ref{thm:lb1}]
\emph{Case 1. $k \ge B$.} In this case, the same hard instance works for both $\PO$ and $\LPO$.
Consider the following sequence of arriving packets: on step $1$, there arrives a packet with $B$ required work followed by a packet with a single required cycle;
on steps $2..B-2$, $B-2$ more packets with a single required processing cycle; on step $B-1$, $B$ packets with a single processing cycle, and then no packets until step $2B-1$, when the sequence is repeated.
Under this sequence of inputs, the queues will work as follows ($\# ALG$ denotes the number of packets processed by $ALG$).

\begin{tabular}{c|c|r|c|r|c}
$t$ & Arriving & $\IB^{\{\PO,\LPO\}}_t$ & \# $\{\PO,\LPO\}$ & $\IB^{\OPT}_t$ & \# $\OPT$ \\\hline
$1$ & $\p{1}$ $\p{B}$ & $\p{1}$ $\p{B}$ & $0$ & $\p{1}$ & $1$ \\
$2$ & $\p{1}$ & $\p{1}$ $\p{1}$ $\p{B-1}$ & $0$ & $\p{1}$ & $2$ \\
$3$ & $\p{1}$ & $\p{1}$ $\p{1}$ $\p{1}$ $\p{B-2}$ & $0$ & $\p{1}$ & $3$ \\
\ldots & & \ldots & & \ldots & \\
$B-2$ & $\p{1}$ & $\p{1}$ \ldots $\p{1}$ $\p{2}$ & $0$ & $\p{1}$ & $B-2$ \\
$B-1$ & $\p{1} \times B$ & $\p{1}$ \ldots $\p{1}$ $\p{1}$ & $1$ & $\p{1}$ \ldots $\p{1}$ $\p{1}$ & $B-1$ \\
\ldots & & \ldots & & \ldots & \\
$2B-1$ &  &  & $B$ &  & $2B-2$ \\
\end{tabular}

Thus, at the end of this sequence $\PO$ has processed $B$ packets, while $\OPT$ has processed $2B-2$,
and the sequence repeats itself, making this ratio asymptotic.

\emph{Case 2.1. $k < B$, algorithm $\PO$.} In this case, we need to refine the previous construction; for simplicity, assume that $k\ll B\to\infty$,
and everything divides everything.
\begin{enumerate}
\item On step~1, there arrive $(1-\alpha)B$ packets of required work $k$ followed by $\alpha B$ packets with required work $1$ ($\alpha$ is a constant to be determined
later). $\PO$ accepts all packets, while $\OPT$ rejects packets with required work $k$ and only accepts packets with required work $1$.
\item On step~$\alpha B$, $\OPT$'s queue becomes empty, while $\PO$ has processed $\frac{\alpha B}{k}$ packets, so it has $\frac{\alpha B}{k}$ free spaces in the queue.
Thus, there arrive $\frac{\alpha B}{k}$ new packets of required work $1$.
\item On step~$\alpha B(1+\frac1k)$, $\OPT$'s queue is empty again, and there arrive $\frac{\alpha B}{k^2}$ new packets of required work $1$.
\item ...
\item When $\PO$ is out of packets with $k$ processing cycles, its queue is full of packets with $1$ processing cycle, and $\OPT$'s queue is empty. At this point, there arrive $B$ new
packets with a single processing cycle, they are processed, and the entire sequence is repeated.
\end{enumerate}
In order for this sequence to work, we need to have
$$\alpha B\left(1+\frac1k+\frac1{k^2}+\ldots\right) = k\left(1-\alpha\right)B.$$
Solving for $\alpha$, we get $\alpha = 1-\frac 1k$. During the sequence, $\OPT$ has processed $\alpha B\left(1+\frac1k+\frac1{k^2}+\ldots\right) + B=2B$ packets,
while $\PO$ has processed $\left(1-\alpha\right)B+B=\left(1+\frac1k\right)B$ packets, so the competitive ratio is $\frac{2}{1+\frac1k}$.
Note that the two competitive ratios, $\frac{2}{1+\frac1k}$ and $2\left(1-\frac1B\right)$, match when $k=B-1$.

\emph{Case 2.2. $k < B$, algorithm $\LPO$.} In this case, we can use an example similar to the previous one, but simpler since there is no extra profit to be had from an iterative construction.
\begin{enumerate}
\item On step~1, there arrive $(1-\alpha)B$ packets with $k$ processing cycles followed by $\alpha B$ packets with a single processing cycle ($\alpha$ is a constant to be determined later).
$\LPO$ accepts all packets, while $\OPT$ rejects packets with required work of $k$ and only accepts packets with a single processing cycle.
\item On step~$\alpha B$, $\OPT$'s queue becomes empty, while $\PO$ has processed $\frac{\alpha B}{k}$ packets, so it has $\frac{\alpha B}{k}$ free spaces in the queue.
There arrive $\beta B$ new packets of required work $1$.
\item On step~$\left(\alpha+\beta\right) B$, $\OPT$'s queue is empty again, and $\LPO$'s queue consists of $B$ packets with required work $1$.
At this point, there arrive $B$ new packets with required work $1$, they are processed, and the entire sequence is repeated.
\end{enumerate}

In order for this sequence to work, we need to have
$$\left(\beta + \frac{\alpha+\beta}{k-1}\right) B = \left(1-\alpha\right)B,$$
and $\OPT$ has processed $(\alpha+\beta)B$ extra packets, and from this equation we get $\alpha+\beta = \left(1+\frac1{k-1}\right)^{-1}$.
During the sequence, $\OPT$ has processed $B\left(1+\left(1+\frac1{k-1}\right)^{-1}\right)$ packets, and $\LPO$ has processed $B$ packets,
yielding the necessary bound.
\qed\end{proof}

\begin{proof}[of Theorem~\ref{thm:po-lower-general}]
We proceed by induction on $B$. For the induction base, we begin with the basic construction that works for $k=\Omega(B^2)$.

\begin{lemma}
\label{thm:po-lower-three}
For $k\ge (B-1)(B-2)$, the competitive ratio of $\PO$ is at least $\frac{3B}{B+1}$;
for $\LPO$, the competitive ratio is at least exactly $3$.
\end{lemma}
\begin{proof}
This time, we begin with the following buffer state:
$$\p{1}\ \p{2}\ \p{3}\ \p{4}\ \ldots\ \p{B-1}\ \p{(B-1)(B-2)}.$$
Over the next $(B-1)(B-2)$ steps, $\PO$ ($\LPO$) keeps processing the first packet, while $\OPT$, dropping the first packet, processes all the rest
(their sizes sum up to the size of the first one). Thus, after $(B-1)(B-2)$ steps $\OPT$'s queue is empty, and $\PO$'s ($\LPO$'s) queue looks like
$$\p{\ \vphantom{C} }\ \p{1}\ \p{2}\ \p{3}\ \ldots\ \p{B-1}.$$
Over the next $B$ steps, $B$ packets of size $1$ arrive in the system. On each step, $\PO$ ($\LPO$) drops the packet from the head of the queue since
it is the largest one, while $\OPT$ keeps processing packets as they arrive.

Thus, at the end of $(B-1)(B-2) + B$ steps, $\PO$ ($\LPO$) has a queue full of $\p{1}$'s and OPT has an empty queue; moreover, $\PO$ ($\LPO$) has processed only one packet (zero packets), while OPT has processed $2B$ packets. Now, for the case of unlimited size incoming burst we have $B$ packets of size $1$ arriving, and after that they are processed and the sequence is repeated, so $\PO$ ($\LPO$) processes $B+1$ packets ($B$ packets) and $\OPT$ processes $3B$ packets per iteration.
\qed
\end{proof}

If $k$ grows further, we can iterate upon this construction to get better bounds.
For the induction step, suppose that we have already proven a lower bound of $n - O(\frac1B)$ ,
and the construction requires maximal required work per packet less than $S=\Omega(B^{n-1})$.

Let us now use the construction from Lemma~\ref{thm:po-lower-three}, but add $S$ to every packet's required work and, consequently, $S(B-1)$ to the first packet's required work:
$$\p{1+S}\ \p{2+S}\ \p{3+S}\ \p{4+S}\ \ldots\ \p{B-1+S}\ \p{(B-1)(B-2+S)}.$$
At first (for the first $(B-1)(B-2+S)$ steps), this works exactly like the previous construction: $\OPT$ processes all packets except the first
while $\PO$ ($\LPO$) is processing the first packet. After that, $\OPT$'s queue is empty, and $\PO$'s ($\LPO$'s) queue is
$$\p{\ \vphantom{S}}\ \p{1+S}\ \p{2+S}\ \p{3+S}\ \ldots\ \p{B-1+S}.$$
Now we can add packets from the previous construction (one by one in the unit-size burst case or all at once),
and $\OPT$ will just take them into its queue, while $\PO$ ($\LPO$) will replace all existing packets from its queue with new ones.
Thus, we arrive at the beginning of the previous construction, but this time, $\PO$ ($\LPO$) has already processed one packet
and $\OPT$ has already processed $B-1$ packets.

This completes the proof of Theorem~\ref{thm:po-lower-general}.\qed
\end{proof}

\end{appendix}

\end{document}